\documentclass{article}

\usepackage{ourStyle}
\usepackage{times}
\usepackage{caption}
\usepackage{subcaption}
\usepackage{graphicx} 

\usepackage{natbib}
\usepackage{algorithm}
\usepackage{algorithmic}
\usepackage{hyperref}

\usepackage[accepted]{icml2015} 

\graphicspath{{./images/}}

\icmltitlerunning{Unwrapping ADMM: Efficient Distributed Computing via Transpose Reduction}

\begin{document} 

\twocolumn[
\icmltitle{Unwrapping ADMM: Efficient Distributed Computing via
Transpose Reduction}

\icmlauthor{Thomas Goldstein}{tomg@cs.umd.edu}
\icmladdress{University of Maryland, College Park}
\icmlauthor{Gavin Taylor}{taylor@usna.edu}
\icmlauthor{Kawika Barabin}{}
\icmlauthor{Kent Sayre}{}
\icmladdress{United States Naval Academy, Annapolis }

\icmlkeywords{ADMM, consensus, unwrapped, transpose reduction, support vector machine, lasso}

\vskip 0.3in
]

 \sloppy  

\begin{abstract} 
 Recent approaches to distributed model fitting rely heavily on consensus
 ADMM, where each node solves small sub-problems using only local data.  We
 propose iterative methods that solve {\em global} sub-problems over an
 entire distributed dataset.  This is possible using transpose reduction
 strategies that allow a single node to solve least-squares over massive
 datasets without putting all the data in one place.  This results in simple
 iterative methods that avoid the expensive inner loops required for consensus
 methods.   To demonstrate the efficiency of this approach, we fit linear
 classifiers and sparse linear models to datasets over 5 Tb in size using a
 distributed implementation with over 7000 cores in far less time than
 previous approaches.
\end{abstract}

\section{Introduction}

We study optimization routines for problems of the form
\eqn{single_term}{ \minimize f(Dx),}
where $D\in \reals^{m\times n}$ is a (large) data matrix and $f$ is a
convex function.  We are particularly interested in the case that $D$ is
stored in a distributed way across $N$ nodes of a network or cluster.  In this
case, the matrix $D = (D_1^T, D_2^T,\cdots,D_N^T)^T$ is a vertical stack of
sub-matrices, each of which is stored on a node.  If the function $f$
decomposes across nodes as well, then problem  \eqref{single_term} takes the
form
\eqn{mult_term}{ \minimize \sum_{i \le N}f_i(D_ix),}
where the summation is over the $N$ nodes.  Problems of this form include
logistic regression, support vector machines, lasso, and virtually all
generalized linear models \cite{HH01}.
 
Most distributed solvers for equation \eqref{mult_term}, such as ADMM, are
built on the assumption that one cannot solve global optimization problems
involving the entire matrix $D.$  Rather, each node alternates between solving
small sub-problems involving only local data, and then exchanging information
with other nodes.
 
This work considers methods that solve {\em global} optimization problems over
the entire distributed dataset on each iteration.  This is possible using {\em
transpose reduction methods}.  Such schemes exploit the following simple
observation: when $D$ has many more rows than columns, the matrix $D^TD$ is
considerably smaller than $D.$  The availability of $D^TD$ enables a single
node to solve least-squares problems involving the entire data matrix $D.$
Furthermore, in many applications it is possible (and efficient) to compute
$D^TD$ in a distributed way. This allows our approach to solve extremely large
optimization problems much faster than the current state-of-the art.  We
support this conclusion both with theoretical results in Section
\ref{sec:theory} and with experimental results in Section
\ref{sec:experiments}.

\section{Related Work}
Simple first-order methods (methods related to gradient descent) have been
available for smooth minimization problems for decades  \cite{TBA86,RN04}.
However, slow performance of gradient schemes for poorly conditioned problems
and slow (sublinear) convergence of sub-gradient and stochastic methods for
non-differentiable problems has led many data scientists to consider splitting
schemes.

Much recent work on solvers for formulation \eqref{mult_term} has focused on the
Alternating Direction Method of Multipliers \cite{GM75,GL89,GO08}, which has
become a staple of the distributed computing and image processing literature.
The authors of \cite{BPCPE10} propose using ADMM for distributed model fitting
using the ``consensus'' formulation.    Consensus ADMM has additionally been
studied for distributed model fitting \cite{EZDV11}, support vector machines
\cite{FCG10}, and numerous domain-specific applications \cite{CPB08, HTY12}.
Many variations of ADMM have subsequently been proposed, including specialized
variants for decentralized systems \cite{MXAP13}, asynchronous updates
\cite{ZK14, OHTG13},  inexact solutions to subproblems \cite{CHW15}, and
online/stochastic updates \cite{OHTG13}.

\section{Background}
The alternating direction method of multipliers (ADMM) is a general method for
solving the problem
\aln{\splt{ \label{two_term}
 \minimize& g(x)+h(y)   \\ 
 \st &Ax+By=0.
}}
The ADMM enables each term of problem \eqref{two_term} to be addressed
separately.  The algorithm in its simplest form begins with estimated
solutions $x^0,$ $z^0,$ and a Lagrange multiplier $\lambda^0.$  ADMM then
generates the following iterates:
\aln{\splt{ \label{admm}
 	x\kp &= \argmin_x g(x)+ \tot\|Ax+By^k+ \lambda^k  \|^2\\
  y\kp &= \argmin_y h(y)+ \tot\|Ax\kp+By+ \lambda^k  \|^2\\
  \lambda\kp & = \lambda^k+Ax\kp+By\kp
}}

   with some positive stepsize parameter $\tau>0.$
   
Disparate formulations are achieved by different $A,$ $B,$ $f,$ and $g$.  For example, consensus ADMM \cite{BPCPE10} addresses the problem 
\aln{
  \splt{ \label{many_term}
  \minimize& \sum_i f_i(x_i)   \\ 
  \st & x_i = z \text{ for all } i
  }
}
which corresponds to problem \eqref{two_term} with $B=(I,I,\cdots I)^T,$
$A=I,$ $f(x) =   \sum_i f_i(x_i),$ and $g=0.$  Rather than solving a single
global problem, the consensus ADMM solves many small problems in parallel.  On
each iteration of consensus ADMM, node $i$ performs the update
\eqn{con_ls}{x_i\kp = \argmin_{x_i} f_i(x_i)+\tot\|x_i-z\|^2.}
The shared global variable $z$ is then updated by the central server. Finally, by updating the Lagrange multipliers $\{\lambda_i\}$, the solutions to each sub-problem are forced to be progressively more similar on each iteration.

\section{Transpose Reduction Made Easy: Regularized Least-Squares}
\label{sec:lasso}
Transpose reduction is most easily described for regularized least-squares
problems.  Consider the distributed solution of the problem 
\eqn{rls}{\minimize J(x)+\half\|Dx-b\|^2}
for some penalty term $J.$  When $J(x) = \mu |x|$ for some scalar $\mu,$ this
becomes the lasso regression \cite{Tibshirani94}.   Typical consensus solvers
for problem \eqref{rls} require each node to compute the solution to
equation \eqref{con_ls}, which is here given by
\mlts{
 x_i\kp = \argmin_{x_i} \half\|D_ix_i-b_i\|^2+\tot\|x_i-z^k\|^2 \\
        = (D_i^TD_i+\tau I)\inv(D_i^Tn+\tau z).
}
During the setup phase for consensus ADMM, each node forms the matrix
$D^T_iD_i,$ and then computes and caches the inverse (or equivalently the
factorization) of $(D_i^TD_i+\tau I).$ 

Alternatively, transpose reduction can solve distributed ADMM on a {\em
single} machine without moving the entire matrix $D$ to one place.  Using the
simple identity 
\begin{align*}
  \half\|Dx-b\|^2 =& \half\la Dx-b,Dx-b\ra \\
  =& \half x^T(D^TD)x-x^TD^Tb+\half \|b\|^2 
\end{align*}
we can replace problem \eqref{rls} with the equivalent problem
\eqn{rls2}{\minimize J(x)+\half x^T(D^TD)x-x^TD^Tb.}
To solve problem \eqref{rls2}, the central server needs only the matrix
$D^TD$ and the vector $D^Tb.$  When $D$ is a large ``tall'' matrix $D\in
\reals^{m\times n},$ with $n\ll m,$ $D^TD$ has only $n^2$ (rather than $nm$)
entries, which is practical to store on a single server.  Furthermore, because
$$D^TD = \sum_i D_i^TD_i, \text{ and } D^Tb = \sum_i D^T_ib_i$$
the matrix $D^TD$ can be formed by having each server compute is own local
$D_i^TD_i,$ and then aggregating/reducing the results on a central server.
  
Once $D^TD$ and $D^Tb$ have been computed in the cloud and cached on a central
server, the global problem is solved on a {\em single} node.  This is done using
either a single-node ADMM method for small dense lasso (see \cite{BPCPE10}
Section 6.4) or a forward-backward (proximal) splitting method
\cite{GSB:2014}.  The latter approach only requires the gradient of $\half
x^T(D^TD)x-x^TD^Tb,$ which is given by $D^TDx-D^Tb.$
 
\section{Unwrapping ADMM:  Transpose Reduction for General Problems}
Transpose reduction can be exploited for more complex problems using
ADMM.  On each iteration of the proposed method, least-squares problems are
solved over the entire distributed dataset.   In contrast, each step of consensus ADMM relies
on sub-problems involving only small subsets of the data.    
 
We aim to solve problem \eqref{single_term}.  We begin by ``unwrapping'' the
matrix $D$; we remove it from $f$ using the formulation 
\aln{
  \splt{
    \label{split}
    \minimize&  f(y)   \\ 
    \st &  y= Dx. 
  }
}
Applying the ADMM with $A=D,$ $B=-I,$ $h=f,$ and $g=0$ yields Algorithm
\ref{alg:unwrap}.
\begin{algorithm}
  \begin{algorithmic}[1]
    \STATE Choose initial $x^0,\,y^0,\,\lambda^0,$ and $\tau>0$
    \WHILE {not converged}
      \STATE $x\kp = \argmin_x \|Dx-y^k+\lambda^k \|^2 = D^+(y^k-\lambda^k)$ 
    \STATE $y\kp = \argmin_y f(y)+\tot\|Dx\kp-y+\lambda^k \|^2$ 
    \STATE $\lambda\kp = \lambda^k+Dx\kp-y\kp$
    \ENDWHILE
  \end{algorithmic} 
  \caption{Unwrapped ADMM}  
  \label{alg:unwrap}
\end{algorithm}

We will now examine the steps in Algorithm \ref{alg:unwrap}.  The $x$ update
requires the solution of a global least squares problem over the {\em entire}
dataset.  The solution requires the pseudoinverse of $D,$ given by
$$D^+=(D^TD)^{-1}D^T.$$   

The $y$ update can be represented using the {\em proximal} mapping of $f,$
which is given by $$\prox_f(z, \delta) =  \argmin_y f(y)+\frac{1}{2\delta}\|y-
z\|^2.$$ Using this notation,  Line 4 of Algorithm \ref{alg:unwrap} is written
$\prox_f(Dx\kp+\lambda^k, \tau\inv).$ Provided $f$ is decomposable, the
minimization in Line 4 is coordinate-wise decoupled.  Each coordinate of
$y\kp$ is computed with either an analytical solution, or using a simple
1-dimensional lookup table of pre-computed solutions.

\subsection{Distributed Implementation}
When $D = (D_1^T, D_2^T,\cdots,D_N^T)^T$  is large and distributed over $N$
nodes, no node has access to the entire matrix $D$ to solve the global least
squares problem for $x\kp$.  This is where we exploit the transpose reduction.
    
We can decompose the rows of $y=(y_1^T,y_2^T, \cdots, y_N^T)^T$ and
$\lambda=(\lambda_1^T,\lambda_2^T, \cdots, \lambda_N^T)^T.$ The constraint in
formulation \eqref{split} now becomes $y_i=D_ix,$ and the least-squares $x$
update in Algorithm \ref{alg:unwrap} becomes    
\eqn{ls}{x\kp = D^+(y^k-\lambda^k) = W\sum_i D_i(y^k_i-\lambda^k_i),}
where $W = (\sum_i D_i^TD_i)\inv.$ Each vector $D_i(y^k_i-\lambda^k_i)$ can be
computed locally on node $i$.  Multiplication by $W$ must take place on a
central server.
    
Note the massive dimensionality reduction that takes place when $D^TD=\sum_i
D_i^TD_i$ is formed. For a data matrix $D\in \reals^{m\times n}$ with $n\ll
m$,  the Gram matrix $D^TD\in \reals^{n\times n}$ is  small, even in
the case that $D$ is far too large to store on the central server.   The complete distributed method is listed in Algorithm \ref{alg:unwrap_dist}.

\begin{algorithm}
  \begin{algorithmic}[1]
    \STATE Choose global $x^0,\, \tau,$ and local $\{y^0_i\},\,\{\lambda^0_i\}$ on each node $i$ 
    \STATE All nodes:  $W_i = D_i^TD_i,$ $\forall i$ 
    \STATE Central node:  $W= (\sum_i W_i)\inv$ 
    \WHILE {not converged}
      \STATE All nodes:  $d_i^k = D_i^T(y^k_i-\lambda^k_i), $ $\forall i$  
      \STATE Central Node:  $x\kp = W\sum_i d_i^k$
      \STATE All nodes:\\
        \begin{varwidth}[t]{\linewidth}
         $ y\kp_i = \argmin_{y_i} f_i(y_i)+\tot\|D_ix\kp-y_i+\lambda^k_i \|^2 $\par
               \hskip.81cm $=\prox_{f_i}(D_ix\kp+\lambda^k_i,\tau\inv),\forall i$
        \end{varwidth}
      \STATE All nodes:  $\lambda\kp_i = \lambda^k_i+D_ix\kp-y\kp_i$
    \ENDWHILE
  \end{algorithmic} 
  \caption{Unwrapped ADMM with Transpose Reduction}  
  \label{alg:unwrap_dist}
\end{algorithm}

\section{Application: Fitting linear classifiers}

\subsection{Logistic regression}
A logistic classifier maps a feature vector $d$ to probabilities using a
mapping of the form $d \mapsto \frac{e^{x\cdot d}}{1+e^{x\cdot d}}$ for some
row vector $x$ of weights.  Given a matrix $D\in \reals^{m\times n}$ of
feature vectors, and a vector $l\in \reals^m$ of labels, the regression
weights are found by solving
\vspace{-1mm} 
\mlt{
\minimize_x \sum_{k=1}^m \log(1+\exp(-l_k(D_kx)))  \\= f_{lr}(Dx),  \label{log}
}

where $ f_{lr}$ is the logistic regression loss (negative
log-likelihood) function
$$f_{lr}(z) = \sum_{k=1}^m \log(1+\exp(-l_kz_k)).$$ 
   
To apply unwrapped ADMM, we must evaluate the proximal
operator $$\prox_{f_{lr}}(z, \tau\inv) =  \argmin_y f_{lr}(z)+\tot\|y-
z\|^2.$$ This is a convex, 1-dimensional optimization problem, and the
solution is easily computed with a few iterations of Newton's method.
However, because the proximal operator depends on only one input variable, we
suggest forming a lookup table of pre-computed solutions. 
   
\subsection{Support Vector Machine}
A common formulation of the support vector machine (SVM) solves
  \eqn{svm}{\minimize \half \|x\|^2 + Ch(Dx)}
where $C$ is a regularization parameter, and $h$ is a simple ``hinge loss''
function given by
$$h(z)  =   \sum_{k=1}^M \max\{1-l_kz_k,0\}.$$ 
The proximal mapping of $h$ has the form
  $$\prox_h(z,\delta)_k = z_k+l_k\max\{\min\{1-l_kz_k,\delta\},0\}.$$
Using this proximal operator, the solution to the $y$ update in Algorithm \ref{alg:unwrap} is simply 
  $$y\kp = \prox_h\left( Dx\kp+\lambda^k,\frac{C}{\tau}\right).$$
  Note that this algorithm is much simpler than the consensus implementation of SVM, which requires each node to solve the sub-problem 
  \eqn{svm_prox}{\minimize Ch(Dx)+\tot\|x-y\|^2.}
  Despite the similarity of this problem to the original SVM \eqref{svm}, this problem form is {\em not} supported by available SVM solvers such as LIBSVM \cite{CL11} and others, and thus requires a custom solver (see Section \ref{sec:svm_dual} in the Appendix).

\section{Handling Sparsity}
Variable-selection methods use $\ell_1$ regularization to obtain sparse
solutions. For high dimensional problems, $\ell_1$ regularization is commonplace to avoid
overfitting.  Sparse model fitting problems have the form
\eqn{l1_reg}{ \minimize \mu |x| + f(Dx)}
for some regularization parameter $\mu>0.$  Sparse problems can be trivially reduced to the form \eqref{single_term} by defining 
     $$ \widehat D =  \mat{I\\D}
     ,  \qquad\hat f(z)_k = \css{
            \mu |z_k|, \text{ for } 1\le k \le n\\
            f_k(z_k) , \text{ for } k > n
          }$$
  and then minimizing  $\hat f(\widehat D x).$  In the distributed setting where $f$ and $D$ are computed/stored across $N$ servers, the problem now has the form \eqref{mult_term} with
  \eqn{mult_term_sparse}{ \minimize \sum_{i \le {N+1}}f_i(D_ix)}
  where $f_{N+1}(z) = \mu|z|,$ and $D_{N+1}  = I.$
  
  \subsection{Splitting Over Columns}
When the matrix $D$ is extremely wide ($m \ll n$), it often happens that each
server stores a subset of columns of $D$ rather than rows.  Fortunately, such
problems can be handled by solving the {\em dual} of the original problem.
The dual of the sparse problem \eqref{l1_reg} is given by 
\eqn{single_dual}{\minimize_\alpha f^*(\alpha) \st \|D^T\alpha\|_\infty \le \mu}
where $f^*$ denotes the Fenchel conjugate \cite{BV04} of $f.$  For example the
dual of the lasso problem is simply
$$\minimize_\alpha \half\|\alpha+b\|^2 \st \|D^T\alpha\|_\infty \le \mu.$$
Problem \eqref{single_dual} then reduces to the form \eqref{single_term} with
$$ \widehat D =  \left(\begin{array}{c}I\\D^T\end{array}\right),
  \,\, \hat f(z)_k =\css{
             \half\|z_k+b_k\|^2, \text{ for } 1\le k \le m\\
            \mathcal{X}(z_k) , \text{ for } k > m
          }$$
where  $ \mathcal{X}(z)$ is the characteristic function of the $\ell_\infty$ ball of radius $\mu.$  The function $ \mathcal{X}(z)$ is infinite when $|z_i|>\mu$ for some $i,$ and zero otherwise.
The unwrapped ADMM for this problem requires the formation of $D_iD^T_i$ on
each server, rather than $D^T_iD_i.$

\section{Convergence Theory}
\label{sec:theory}
The convergence of ADMM is in general well understood.  Classical results guarantee convergence using monotone operator theory \cite{EB92}, however they do not provide convergence rates.  A more recent result due to He and Yuan \cite{HY12:admm}, provides a global convergence rate for the method.  Theorem \ref{thrm} below presents this result, as adapted to the problem form \eqref{two_term}.

\begin{theorem} (He and Yuan)  Consider the iterates generated by the ADMM method \ref{admm}.  If $f$ and $g$ are proper convex functions, then 
\mlt{
 \|B(y\kp-y^k)\|^2+\|Ax\kp+By\kp\|^2  \\ 
 \le  \frac{\|B(y^0-y\opt)\|^2+\|\lambda^0-\lambda\opt\|^2}{k+1}
}
where $y\opt$ is an optimal solution for \eqref{two_term} and $\lambda\opt$ is the corresponding optimal Lagrange multiplier.
\label{thrm} 
\end{theorem}
It was observed in \cite{BPCPE10} that $\|Ax\kp+By\kp\|^2$ and
$\|A^TB(y\kp-y^k)\|^2$ are measures of  primal and dual infeasibility,
respectively.  Thus, Theorem \ref{thrm} guarantees the iterates of ADMM
approach primal and dual feasibility (and thus optimality) as $k\to \infty,$
and that the infeasibility errors decrease with rate $O(1/k)$ in the worst
case.
We can specialize this result to the case of unwrapped ADMM to obtain Corollary \ref{cor}.
\begin{corollary} If formulation \eqref{single_term} is feasible, then the
iterates of the unwrapped ADMM \ref{alg:unwrap} satisfy 
\mlts{ \|y\kp-y^k\|^2+\|Dx\kp -y\kp\|^2 \\
 \le  \frac{\|y^0-Dx\opt\|^2+\|\lambda^0-\lambda\opt\|^2}{k+1}}
where $y\opt$ is an optimal solution for \eqref{two_term} and $\lambda\opt$ is the corresponding optimal Lagrange multiplier.
\label{cor} 
\end{corollary}

While Theorem \ref{thrm} and Corollary \ref{cor} show convergence in the sense that primal and dual feasibility are attained for large $k$, it seems more natural to ask for a direct measure of optimality with respect to the objective function \eqref{single_term}.  In the case that $f$ is differentiable, we can exploit the simple form of unwrapped ADMM and show the derivative of \eqref{single_term} goes to zero, thus directly showing that $x^k$ is a good approximate minimizer of \eqref{single_term} for large $k$.

\begin{theorem}
If the gradient of $f$ exists and has Lipschitz constant $L(\nabla f),$ then Algorithm \ref{alg:unwrap} shrinks the gradient of the objective function \eqref{single_term} with global rate
\mlts{
\| \nabla \{f(Dx^k)\}\|^2=\| D^T\nabla f(Dx^k)\}\|^2  \\
   \le C\frac{\|y^0-Dx\opt\|^2+\|\lambda^0-\lambda\opt\|^2}{k}
 }
where $ C = (L(\nabla f)+\tau)^2\rho(D^TD)$ is a constant and $\rho(D^TD)$
denotes the spectral radius of $D^TD.$
\label{thrm2} 
\end{theorem}
\begin{proof}
We begin by writing the optimality condition for the $x$-update in Algorithm \ref{alg:unwrap}:
\mlt{ \label{eq1}D^T(Dx\kp-y^k+\lambda^k) \\ =D^T \lambda\kp+D^T(y\kp-y^k)=0.}
Note we used the definition $\lambda\kp = \lambda^k+Dx\kp-y\kp$ to simplify \eqref{eq1}. 
Similarly, the optimality condition for the $y$-update yields
\mlts{\nabla f(y\kp)+\tau(y\kp-Dx\kp-\lambda^k) \\ = \nabla f(y\kp)-\tau\lambda\kp=0,}
or equivalently $\nabla f(y\kp) = \tau \lambda\kp.$  Combining this with
equation \eqref{eq1} yields
$$D^T\nabla f(y^k)=\tau D^T(y^k-y\kp).$$
We now have
\alns{
\partial_x \{f(Dx^k)\} =& D^T\nabla f(Dx^k) = D^T\nabla f(y^k+Dx^k-y^k) \\
=& D^T\nabla f(y^k+Dx^k-y^k) \\
   &-D^T\nabla f(y^k)+\tau D^T(y^k-y\kp)
}
and so $\|\partial_x \{f(Dx^k)\}\| $ is bounded above by
$$\|D^T\nabla f(y^k+Dx^k-y^k)-D^T\nabla f(y^k)\|  +\tau\|D^T(y^k-y\kp)\| $$
$$ \le  L(\nabla f)\|D^T\|_{op}  \|Dx^k-y^k\|+\tau\|D^T\|_{op} \|y^k-y\kp\|. \label{eq2}
$$
 where $\|D^T\|_{op}$ denotes the operator norm of $D^T$.

Now, by Corollary \eqref{cor}, we know that both  
$\|Dx^k-y^k\|$ and   $\|y^k-y\kp\|$ are bounded above by $\sqrt{\frac{\|y^0-Dx\opt\|^2+\|\lambda^0-\lambda\opt\|^2}{k}}.$  Applying this bound to \eqref{eq2} yields
\alns{
\|\partial_x \{f(Dx^k)\}\| \le&  \sqrt{C\frac{\|y^0-Dx\opt\|^2+\|\lambda^0-\lambda\opt\|^2}{k}}.
 }
We obtain the result by squaring this inequality and noting that $\|D\|_{op}^2 = \rho(D^TD).$
\end{proof}
Note  the logistic regression problem \eqref{log} satisfies the conditions of Theorem \ref{thrm2} with $L(\nabla f) = 1/4.$

Finally, we remark that better convergence rates are possible in special cases.  For example when $f$ is strongly convex, accelerated ADMM obtains $O(1/k^2)$ convergence \cite{GOSB14}, and if we further assume $D$ has full row rank, R-linear convergence is possible \cite{DY12}.

\section{Implementation Details}
We compare transpose reduction methods to consensus ADMM using both synthetic
and empirical data.  We study the transpose reduction scheme for lasso (Section \ref{sec:lasso}) in addition to the unwrapped ADMM (Algorithm \ref{alg:unwrap_dist}) for logistic regression and SVM.

We built a distributed implemention of the transpose reduction methods along with consensus
optimization, and ran all experiments on Armstrong, a 30,000-core Cray
XC30 hosted by the DOD Supercomputing Resource Center.
This allow us to study experiments ranging in size from very small to
extremely large.  

All distributed methods were implemented using MPI.  Stopping conditions for both
transpose reduction and consensus methods were set using the residuals defined
in \cite{BPCPE10} with $\epsilon_{rel}=10^{-3},$ and $\epsilon_{abs}=10^{-6}.$

Many steps were taken to achieve top performance of the consensus optimization
routine.  The authors of \cite{BPCPE10} suggest using a stepsize parameter
$\tau=1;$ however, substantially better performance is achieved by tuning this
parameter.  In our experiments, we tuned the stepsize parameter to achieve
convergence in a minimal number of iterations on a problem instance with
$m=10,000$ data vectors and $n=100$ features per vector, and then scaled the
stepsize parameter up/down to be proportional to $m.$  It was found that this
scaling made the number of iterations nearly invariant to $n$ and $m.$  The
iterative solvers for each local logistic regression/SVM problem were
warm-started using solutions from the previous iteration.   

The logistic regression subproblems were solved using a limited memory BFGS method (with warm start to accelerate performance).  The transpose reduced lasso method (Section \ref{sec:lasso} ) requires a sparse least-squares method to solve the entire lasso problem on a single node.  This was
accomplished using the forward-backward splitting implementation FASTA
\cite{GSB:2014, FASTA:2014}.

Note that the SVM problem \eqref{svm_prox} cannot be solved by conventional
SVM solvers (despite its similarity to the classical SVM formulation).  For
this reason, we built a custom solver using the same coordinate descent
techniques as the well-known solver LIBSVM \cite{CL11}.  By using warm starts
and exploiting the structure of problem \eqref{svm_prox}, our custom method
solves the problem \eqref{svm_prox} dramatically faster than standard solvers
for problem \eqref{svm}.  See Appendix \ref{sec:svm_dual} for details.

\section{Numerical Experiments}
\label{sec:experiments}

\begin{figure*}[ht]
  \centering
  \begin{subfigure}[b]{0.3\textwidth}
    \includegraphics[width=\textwidth]{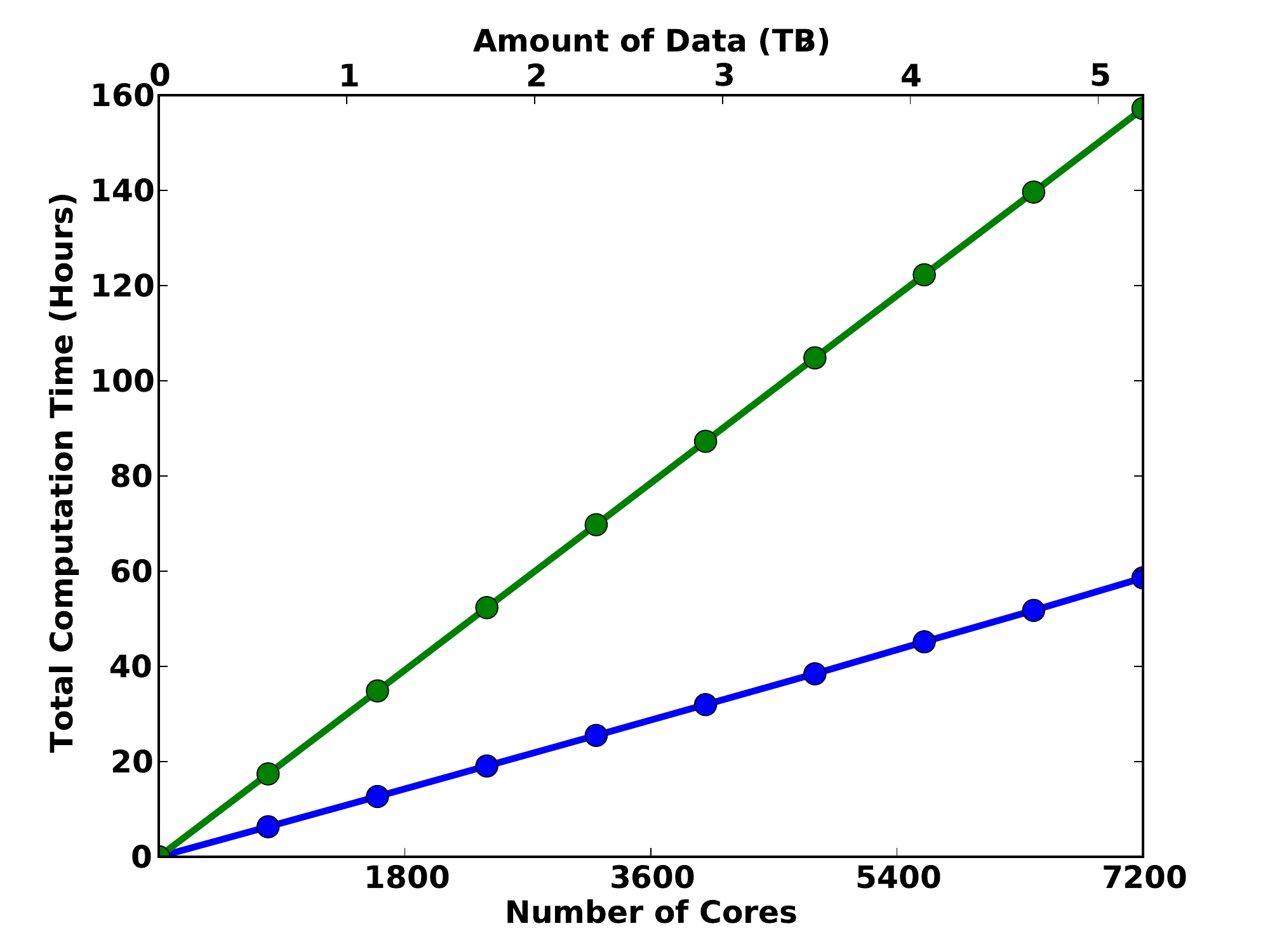}
    \caption{\textbf{Logistic regression with homogeneous data}.  Experiments
    used 50,000 data points of 2,000 features each per computing core.\\ \\}
    \label{fig:tuneNodes}
  \end{subfigure}
  \hfill
  \begin{subfigure}[b]{0.3\textwidth}
    \includegraphics[width=\textwidth]{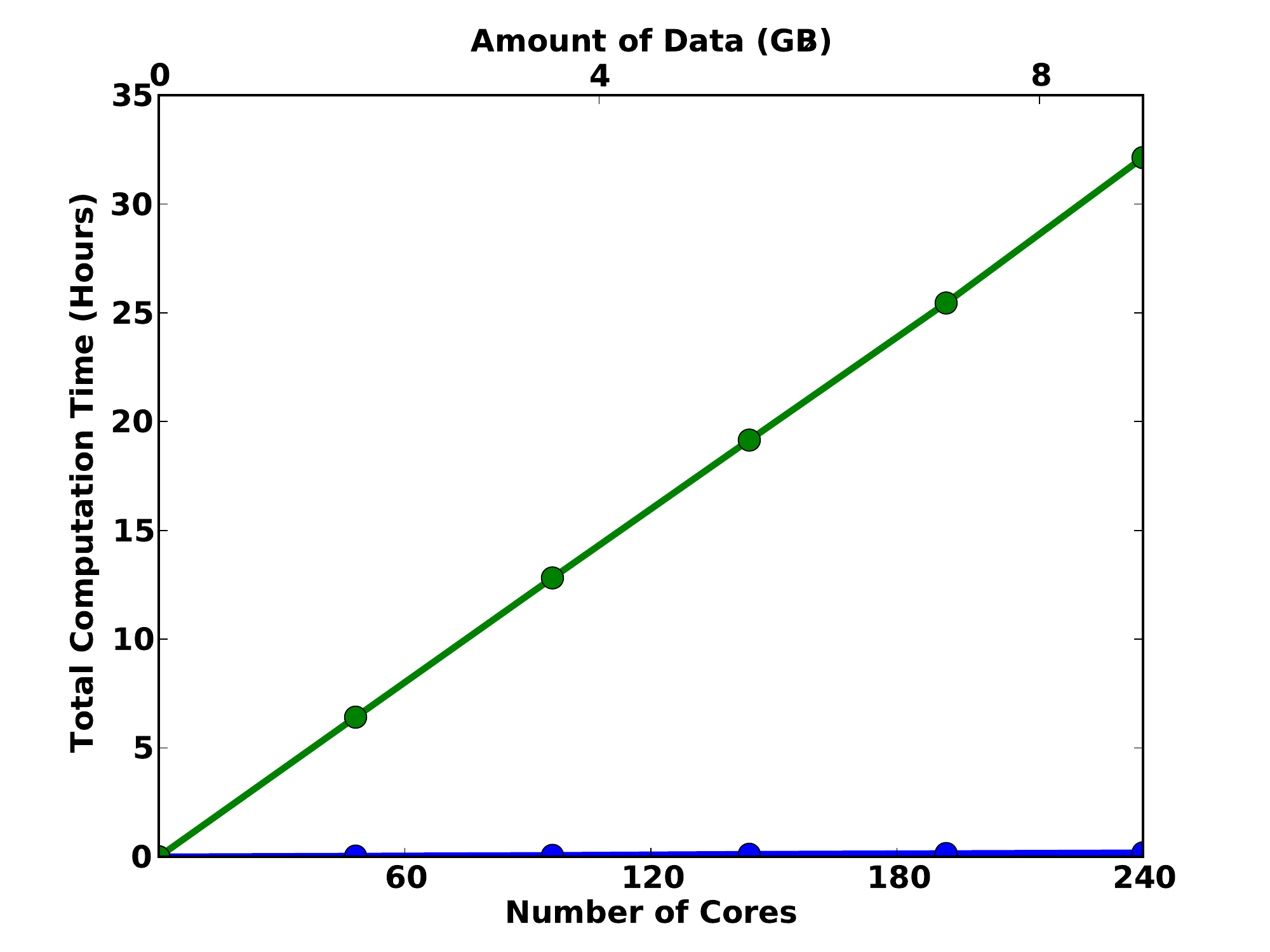}
    \caption{\textbf{SVM with homogeneous data}. Experiments used 50,000 data
      points of 100 features each per computing core. The longest timed run of
      Transpose ADMM took about 11 minutes of total computing time.}
    \label{fig:svm}
  \end{subfigure}
  \hfill
  \begin{subfigure}[b]{0.3\textwidth} 
    \includegraphics[width=\textwidth]{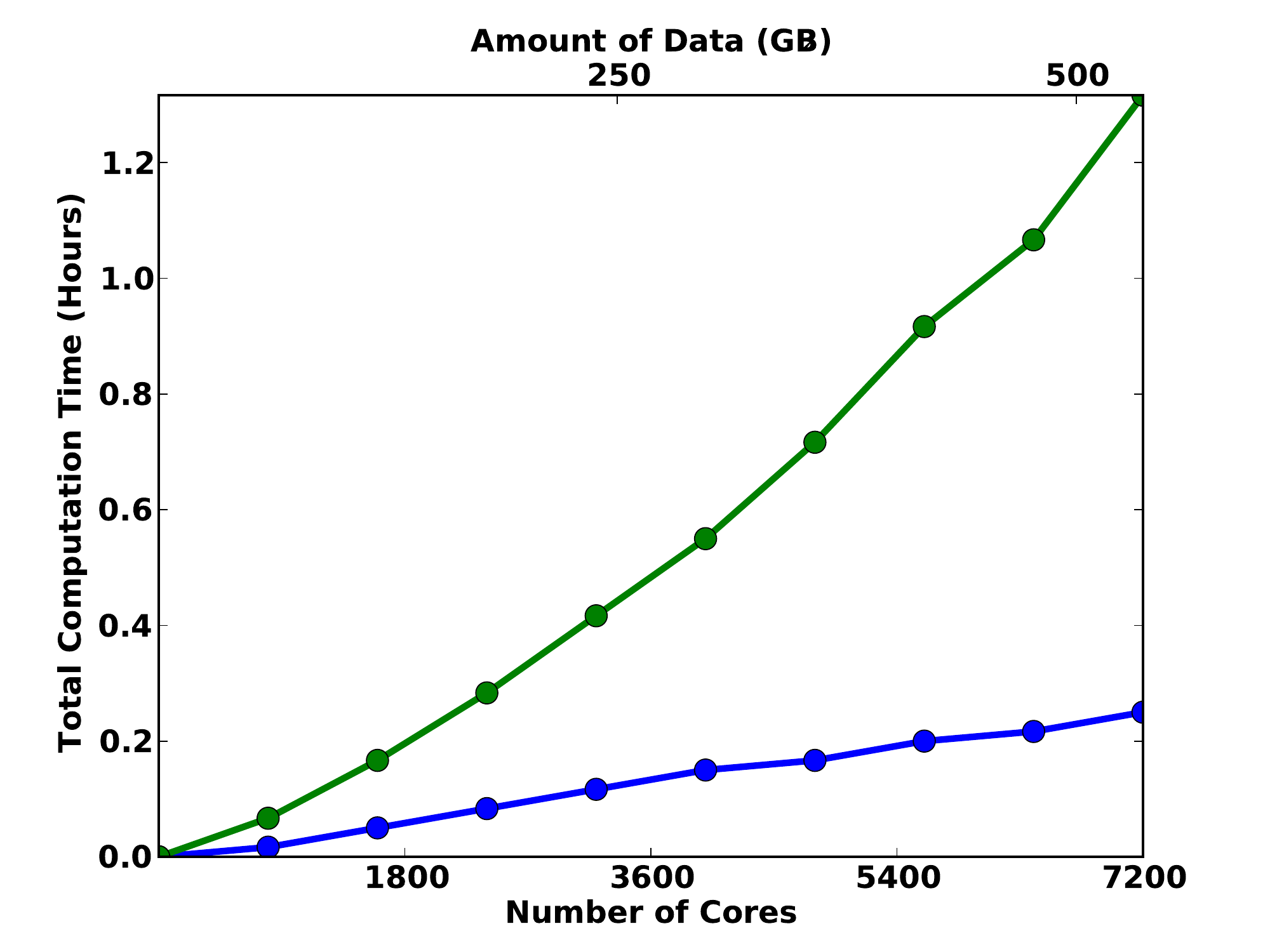}
    \caption{\textbf{Lasso with heterogeneous data}. Experiments used 50,000
    data points of 200 features each per computing core.\\ \\ \\}
    \label{fig:lasso}
  \end{subfigure}
  \vspace{-4mm} 
  \caption{Selected results from timing \textbf{Consensus ADMM (green)} and
    \textbf{Transpose ADMM (blue)} when solving three different optimization
    problems with varying data corpus sizes. In each experiment, every core
    was given an identically-sized subset of the data, so data corpus size and
    number of cores are related linearly.  The top horizontal axis denotes the
    total data corpus size, while the bottom horizontal axis denotes the
    number of computing cores used to process the data.  The vertical axis
    denotes total computing time.
  }
  \label{fig:tuneCores}
\end{figure*}

To illustrate the behavior of transpose reduction methods on various data
corpus sizes and multiple optimization problems, we applied consensus and
transpose solvers to both synthetic and empirical datasets. We recorded both
the {\em total compute time} and the {\em wallclock time}.  We define total
compute time as the total amount of time computing cores spend performing
calculations. Total compute time does not include communication and extra
diagnostic computation not necessary for optimization (such as per-iteration
computation of the objective function); wall time includes all calculation and
communication.

\subsection{Synthetic Data} \label{synth}

Synthetic data was constructed as follows:\\
\textbf{Lasso problems}  We use the same synthetic problems used to study
consensus ADMM in \cite{BPCPE10}.  The data matrix $D$ is a random Gaussian
matrix.  The true solution $x_{true}$ contains 10 active features with unit
magnitude, and the remaining entries are zero.  The $\ell_1$ penalty $\mu$ is
chosen as suggested in  \cite{BPCPE10} --- i.e., the penalty is 10\% the
magnitude of the penalty for which the solution to \eqref{rls} becomes zero.
The observation vector is $b=Dx_{true}+\eta,$ where $\eta$ is a standard
Gaussian noise vector with $\sigma=1.$\\
\textbf{Classification problems} We generated two random Gaussian matrices,
one for each class.  The first class consists of zero-mean Gaussian entries.
The first 5 columns of the second matrix were random Gaussian with mean 1, and
the remaining columns were mean zero.  Note the data classes generated by this
process are not perfectly linearly separable.  The $\ell_1$ penalty was set
using the ``10\%'' rule used in \cite{BPCPE10}.

In addition to these ``standard'' test problem, we also study the important
(and often more realistic) case where data is heterogeneous across nodes.\\
\textbf{Heterogeneous Data} To create data heterogeneity, we chose one random
Gaussian scalar for each node, and added it to the matrix $D_i.$  This has the
effect of making the behavior of the sub-problems on each node different.
Note that, with standard Gaussian test problems, each consensus node is
essentially solving the same problem (i.e., data matrices have perfectly
identical statistical properties), and thus they arrive at a consensus
quickly.  In practical applications, data on different nodes may represent
data from different sources, or else not be identically distributed.  The
behavior of consensus in such situations is radically different than when the
data is identically distributed.  The heterogeneous synthetic data models this
scenario.

We performed experiments using logistic regression, SVM, and Lasso. On
different trial runs, we vary the number of cores used, the length of the
feature vectors, and the number of data points per core.  Three examples of
our results are illustrated in Figure \ref{fig:tuneCores}, while more complete
tables of results appear Appendix
\ref{app:tables}.  In addition, convergence curves for experiments on both
homogeneous and heterogeneous data can be seen in Figures \ref{fig:homoLR7200}
and \ref{fig:heteroLR7200}.

\subsection{Empirical Case Study:  Classifying Guide Stars}

We study the behavior of transpose reduction and consensus methods using the
Second Generation Guide Star Catalog (GCS-II) \cite{lasker2008second}, an
astronomical database containing 950 million stars and other astronomical
objects.  Intensity measurements in different spectral bands are recorded for
each object, in addition to geometric properties such as size and
eccentricity.  The GSC-II also contains a binary classification that labels
every astronomical body as ``star'' or ``not a star.''   We train a sparse
logistic classifier to discern this classification using only spectral and
geometric features.

A data matrix was compiled by selecting the 17 spectral and geometric
measurements reported in the catalog, and also ``interaction features'' made
of all second-order products of the 17 measurements; all features were then
normalized.  After the addition of a bias feature, the resulting matrix has
307 features per object, and requires 1.8 TB of space to store in binary form.  

We ran experiments observing the decrease of the global objective function as
a function of wallclock time; these experiments showed that transpose ADMM
converged far more quickly than consensus.  Convergence curves for
this experiment are shown in Figure \ref{fig:starConvergence}.  

We also experiment with the effect of storing the data matrix across different
numbers of cores (i.e. increasing the number of nodes). These experiments
illustrated that this variable had little effect on the relative performance
of the two optimization strategies; transpose methods remained far more
efficient regardless of the load on individual cores.  Table \ref{fig:gsc_table} reports runtimes with varying
numbers of cores.

Note the strong advantage of transpose reduction over consensus ADMM in Figure
\ref{fig:starConvergence}.  This confirms the results of Section \ref{synth},
where it was observed that transpose reduction methods are particularly
powerful for the heterogeneous data observed in realistic datasets, as opposed
to the identically distributed matrices used in many synthetic experiments.

\begin{figure*}[ht]
  \centering
  \begin{subfigure}[b]{0.3\textwidth}
    \includegraphics[width=\textwidth]{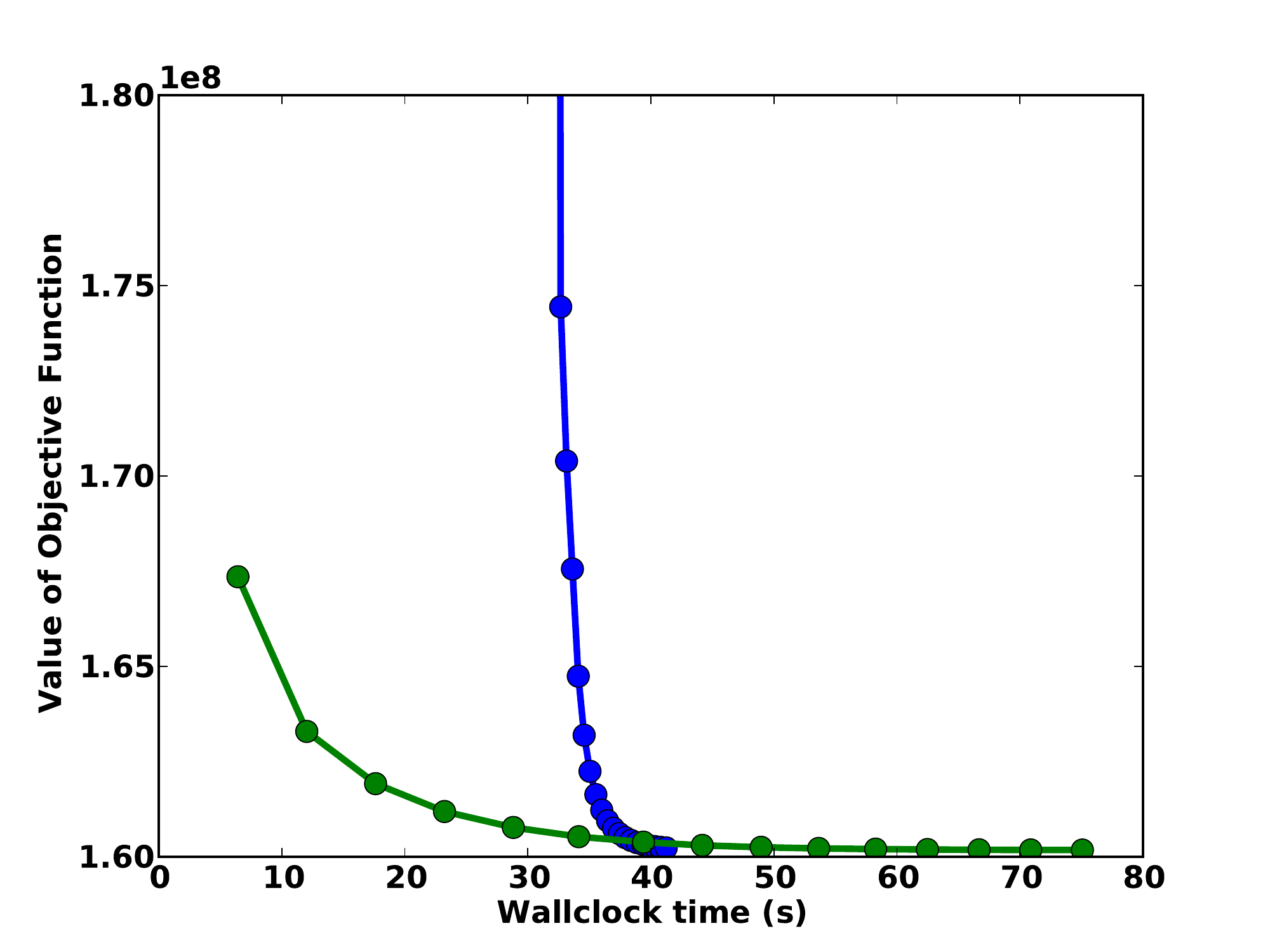}
    \caption{\textbf{Homogeneous data}.  Experiments used 7,200 computing
    cores with 50,000 data points of 2,000 features each. \\}
    \label{fig:homoLR7200}
  \end{subfigure}
  \hfill
  \begin{subfigure}[b]{0.3\textwidth}
    \includegraphics[width=\textwidth]{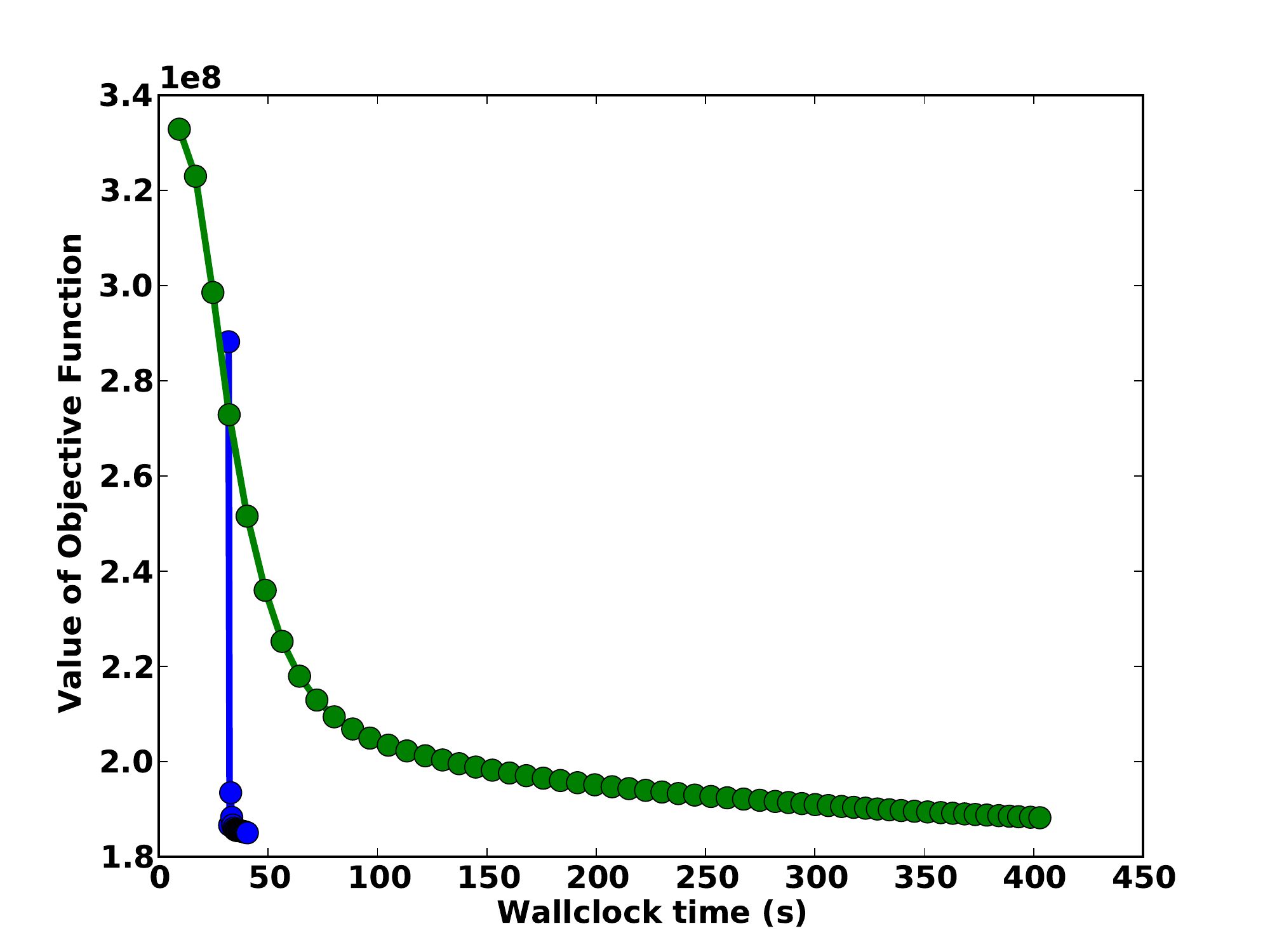}
    \caption{\textbf{Heterogeneous data}.  Experiments used 7,200 computing
    cores with 50,000 data points of 2,000 features each. \\}
    \label{fig:heteroLR7200}
  \end{subfigure}
  \hfill
  \begin{subfigure}[b]{0.3\textwidth}
    \includegraphics[width=\textwidth]{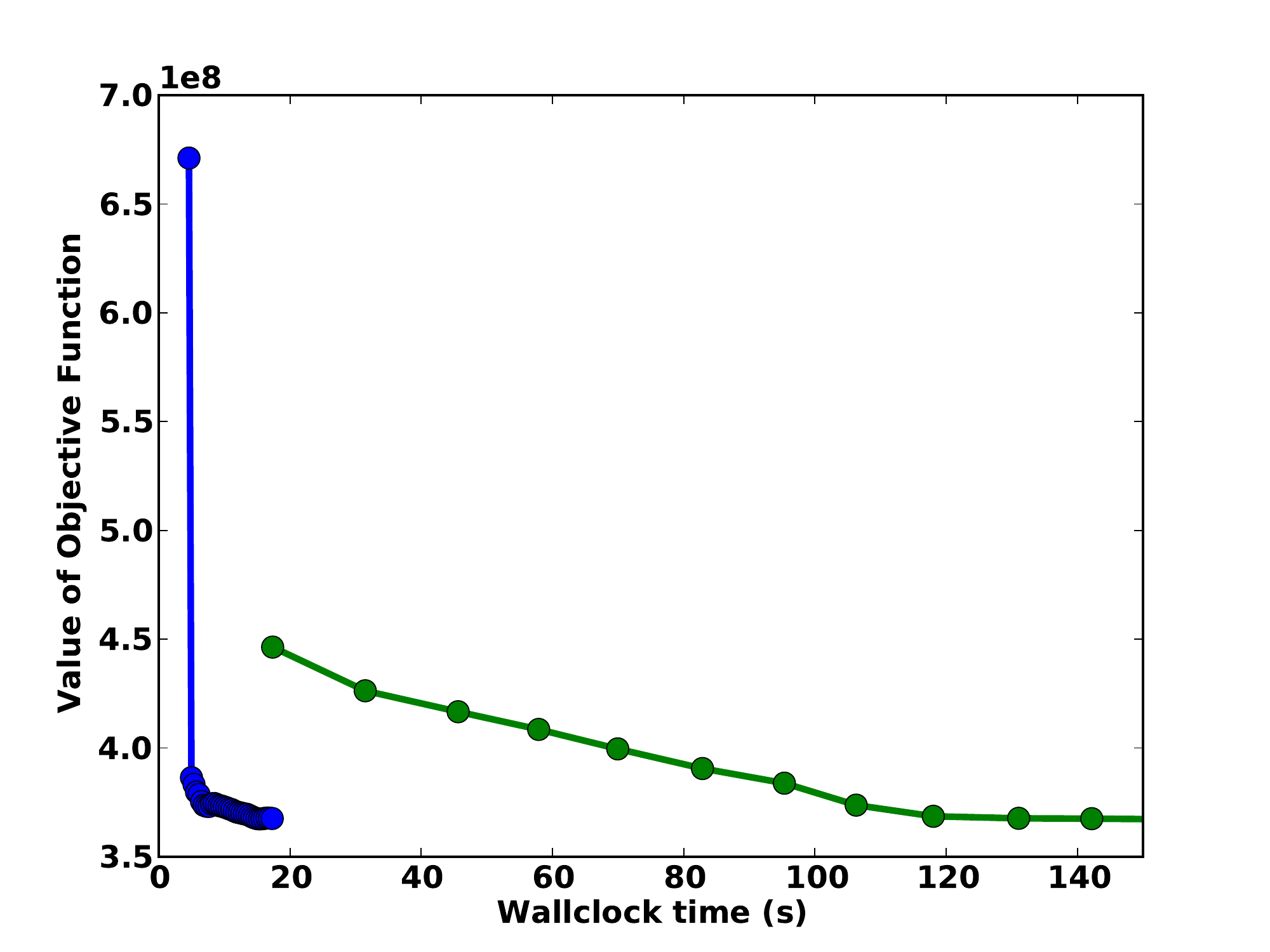}
    \caption{\textbf{Empirical star data}.  Experiments used 2,500 cores to
    classify a data set of 1.8TB. Consensus ADMM did not terminate until 1160
    seconds.}
    \label{fig:starConvergence}
  \end{subfigure}
    \vspace{-4mm} 
  \caption{Value of the logistic regression objective function as a function
    of wallclock time for \textbf{Consensus ADMM (green)} and
    \textbf{Transpose ADMM (blue)}.}
  \label{fig:convergence}
\end{figure*}

 \begin{center}
\begin{table} 
\begin{tabular}{lll|ll}
\hline
Cores &T-Wall &T-Comp &C-Wall &C-Comp\\
\hline
2500 &0:01:06 &11:35:25 &0:24:39 &31d,19:59:13\\
3000 &0:00:49 &12:10:33 &0:21:43 &32d, 2:44:11\\
3500 &0:00:50 &12:17:27 &0:17:01 &30d, 7:56:19\\
4000 &0:00:45 &12:38:24 &0:29:53 &40d, 13:38:19\\
\hline \vspace{-4mm}
\label{tab:star}
\end{tabular}
\caption{Wall clock and Total compute times for the logistic regression
problem on the Second Generation Guide Star Catalog.  Experiments show
different numbers of cores used for the same 1.8TB dataset.  Columns labeled
with ``T-'' denote results for transpose reduction ADMM, and the label ``C-''
denotes consensus ADMM. Times are reported in the format (days,)
hours:minutes:seconds.}  
\label{fig:gsc_table}
\end{table}
\vspace{-8mm}
\end{center}

\section{Discussion}
In all experiments, transpose reduction methods required substantially less computation
time than consensus methods.  As seen in Figure \ref{fig:tuneCores}, Figure
\ref{fig:convergence}, and Appendix
\ref{app:tables}, performance of both transpose reduction and consensus
methods scales nearly linearly with the number of cores and amount of data
used for classification problems.  For the lasso problem (Figure \ref{fig:lasso}), the runtime of the transpose reduction method appears to
grows sub-linearly with number of cores, in contrast to consensus methods that
have a strong dependence on this parameter.  This is largely because the
transpose reduction method does all iteration on a single machine, whereas
consensus optimization requires progressively more communication with larger
numbers of cores.

Note that transpose reduction methods require more startup time for some
problems than consensus methods because the local Gram matrices $D_i^TD_i$
must be sent to the central node, aggregated, and the result inverted; this is
not true for the lasso problem, for which consensus solvers must also invert
a local Gram matrix on each node, though this at least saves startup communication
costs.  This startup time is particularly noticeable when overall solve time
is short, as in Figure \ref{fig:heteroLR7200}.  Note, however, that even for
this problem total computation time and wall time was still shorter with
transpose reduction than with consensus methods.

\subsection{Effect of Heterogeneous Data}
An important property of the proposed transpose reduction methods is that they
solve {\em global} subproblems over the entire data corpus (as opposed to
consensus methods where sub-problems involve a small portion of the
distributed data).  The power of solving global problems is very apparent when
data is heterogeneous across nodes.  Heterogeneity means that data on
different nodes has different statistical properties (as opposed to
homogeneous experiments where every data matrix has identical properties).  

Heterogeneity has a strong effect on consensus methods.  When data is
heterogeneous across nodes it causes local sub-problems to differ, and thus
the nodes have a stronger tendency to ``disagree'' on the solution, taking
longer to reach a consensus.  This effect is illustrated by a comparison
between Figures \ref{fig:homoLR7200} and \ref{fig:heteroLR7200}, where
consensus methods took much longer to converge on heterogeneous data sets,
while the transpose method was not affected.

In contrast, because transpose reduction solves {\em global} sub-problems
across the entire distributed data corpus, it is relatively insensitive to
data heterogeneity across nodes.  Data heterogeneity explains the strong
advantage of transpose reduction on the  GSC-II dataset (Figure
\ref{fig:starConvergence}, Table \ref{fig:gsc_table}), which contains
empirical data and is thus not uniformly distributed.

\subsection{Communication \& Computation}
The transpose reduction ADMM leverages a tradeoff between communication and
computation. When $N$ nodes are used to solve a problem with a distributed
data matrix $D\in \reals^{m\times n},$ each node in consensus ADMM transmits
$x_i\in \reals^n$ to the central server, which totals to $O(Nn)$
communication.  Unwrapped ADMM requires $O(m)$ communication per iteration,
which is often somewhat more.  Despite this difference, transpose reduction
methods are still highly efficient because (a) they require dramatically less
computation, and (b) communication with the server is substantially more
efficient, even when more information is sent. We discuss these two issues in
detail below.

First, unwrapped ADMM requires substantially less computation than consensus
methods.  Consensus ADMM requires inner iterations to solve expensive
sub-problems on each node.  Because logistic regression and SVM can be very
expensive for large data sets, consensus optimization is often CPU bound
rather than communication bound.  In contrast, the sub-problems of unwrapped
ADMM are coordinate-wise separable and generally available in closed form.    

Second, transpose reduction methods stay synchronized better than consensus
ADMM, which makes communication more efficient on synchronous architectures.
The iterative methods used by consensus ADMM for logistic regression and SVM
sub-problems do not terminate at the same time on every machine, especially
when the data is heterogeneous across nodes.  Because all client nodes must
complete their inner iterations before an iteration is complete, most of the
communication overhead of ADMM is caused by calls to MPI routines that block
until all nodes become synchronized.  In contrast, Algorithm
\ref{alg:unwrap_dist} requires the same computations on each server, allowing
nodes to stay synchronized naturally.  

\section{Conclusion}
We introduce transpose reduction ADMM --- an iterative method that solves
model fitting problems using global least-squares subproblems over an entire
distributed dataset.  Numerical experiments using both synthetic and empirical
data suggest that the transpose reduction is substantially more efficient than
consensus methods.   When the distributed dataset has heterogeneous properties
across nodes, the global properties of transpose reduction are particularly
advantageous compared to consensus methods, which solve many small
sub-problems using different datasets.

\bibliographystyle{icml2015}
\bibliography{tom_bibdesk}

\newpage
\onecolumn
\begin{center}\huge{ Appendices}\end{center}
\appendix
\section{SVM Sub-steps for Consensus Optimization} \label{sec:svm_dual}
The consensus SVM requires a solution to the problem \eqref{svm_prox}.
Despise the apparent similarity of this proximal-regularized problem to the
original SVM \eqref{svm}, problem  \eqref{svm_prox} cannot be put into a form
that is solvable to popular solvers for \eqref{svm}.   However, techniques for
the classical SVM problem can be easily adapted to solve \eqref{svm_prox}.

 A common numerical approach to solving \eqref{svm} is the attack its dual, which is
 \eqn{svm_dual}{
 \minimize_{\alpha_i\in [0,C]}  \half\| A^TL\alpha\|^2  - \alpha^T\one 
 \\=  \sum_{i,j}\alpha_i\alpha_jl_il_jA_iA_j^T - \sum_i \alpha_i .
}
Once \eqref{svm_dual} is solved to obtain $\alpha\opt$, the solution to
\eqref{svm} is simply given by $w\opt = L^T\alpha.$  The dual formulation
\eqref{svm_dual} is advantageous because the constraints on $\alpha$ act
separately on each coordinate. The dual is therefore solved efficiently by
coordinate descent, which is the approach used by the popular solver LIBSVM
\cite{CL11}.  This method is particularly powerful when the number of support
vectors in the solution is small, in which case most of the entries in
$\alpha$ assume the value $0$ or $C.$

In the context of consensus ADMM, we must solve 
 \eqn{svm_con}{\minimize \half \|w\|^2 + Ch(Aw,l)+\tot \|w-z\|^2.}
 Following the classical SVM literature, we dualize this problem to obtain
  \eqn{svm_dual_con}{
 \minimize_{\alpha_i\in [0,C]}  \half\| A^TL\alpha\|^2  - \alpha^T\left((1+\tau)\one - \tau Lz\right).
 }
We then solve \eqref{svm_dual_con} for $\alpha\opt,$ and recover the solution via
   $$w\opt = \frac{A^TL\alpha+\tau z}{1+\tau}.$$
 
We solve \eqref{svm_dual_con} using a dual coordinate descent method inspired by
\cite{CL11}.   The implementation has $O(M)$ complexity per iteration. Also
following \cite{CL11} we optimize the convergence by updating coordinates with
the largest residual (derivative) on each pass.
 
Because our solver does not need to handle a ``bias''  variable (in consensus
optimization, only the central server treats the bias variables differently
from other unknowns), and by using a warm start to accelerate solve time
across iterations, our coordinate descent method significantly outperforms
even LIBSVM for each sub-problem.  On a desktop computer with a Core i5
processor, LIBSVM solves the synthetic data test problem with $m=100$
datapoints and $n=200$ features in $3.4$ seconds (excluding ``setup'' time),
as opposed to our custom solver which solves each SVM sub-problem for the
consensus SVM with the same dimensions (on a single processor) in 0.17 seconds
(averaged over all iterations). When $m=10000$ and $n=20,$ LIBSVM requires
over $20$ seconds, while the average solve time for the custom solver embedded
in the consensus method is only  $2.3$ seconds.

\section{Tables of Results}
\label{app:tables}

In the following tables, we use these labels:
\begin{itemize}
  \item N: Number of data points per core
  \item F: Number of features per data point
  \item Cores: Number of compute cores used in computation
  \item Space: Total size of data corpus in GB (truncated at GB)
  \item TWalltime: Walltime for transpose method (truncated at seconds)
  \item TCompute: Total computation time for transpose method (truncated at
    seconds)
  \item CWalltime: Walltime for consensus method (truncated at seconds)
  \item CCompute: Total computation time for consensus method (truncated at
    seconds)
\end{itemize}

\paragraph{Logistic regression with homogeneous data}

\begin{center}
\begin{tabular}{llllllll}
\hline
N &F &Cores &Space(GB) &TWalltime &TCompute &CWalltime &CCompute\\
\hline
50000 &2000 &800 &596 &0:00:53 &6:19:14 &0:01:36 &17:25:18\\
50000 &2000 &1600 &1192 &0:00:58 &12:40:24 &0:01:51 &1 day 10:51:33\\
50000 &2000 &2400 &1788 &0:01:00 &19:05:13 &0:01:52 &2 days 4:21:25\\
50000 &2000 &3200 &2384 &0:01:00 &1 day 1:30:18 &0:01:41 &2 days 21:46:28\\
50000 &2000 &4000 &2980 &0:00:58 &1 day 7:58:24 &0:01:39 &3 days 15:17:51\\
50000 &2000 &4800 &3576 &0:00:58 &1 day 14:27:31 &0:02:31 &4 days 8:49:58\\
50000 &2000 &5600 &4172 &0:01:00 &1 day 21:10:38 &0:02:13 &5 days 2:16:56\\
50000 &2000 &6400 &4768 &0:01:03 &2 days 3:46:42 &0:02:08 &5 days 19:39:40\\
50000 &2000 &7200 &5364 &0:01:21 &2 days 10:36:36 &0:01:47 &6 days 13:12:59\\
\hline
5000 &2000 &4800 &357 &0:00:33 &4:04:11 &0:00:26 &21:01:22\\
10000 &2000 &4800 &715 &0:00:26 &7:51:06 &0:01:22 &1 day 21:24:47\\
15000 &2000 &4800 &1072 &0:00:38 &11:23:22 &0:01:37 &2 days 19:42:30\\
20000 &2000 &4800 &1430 &0:00:42 &15:15:01 &0:01:30 &3 days 19:27:24\\
25000 &2000 &4800 &1788 &0:00:42 &18:59:04 &0:01:48 &4 days 17:24:59\\
30000 &2000 &4800 &2145 &0:00:47 &22:53:25 &0:02:04 &5 days 16:30:28\\
35000 &2000 &4800 &2503 &0:00:57 &1 day 2:43:48 &0:02:46 &6 days 15:10:40\\
40000 &2000 &4800 &2861 &0:00:54 &1 day 6:22:51 &0:02:42 &7 days 14:58:02\\
45000 &2000 &4800 &3218 &0:00:57 &1 day 10:05:17 &0:03:02 &8 days 15:11:42\\
50000 &2000 &4800 &3576 &0:01:02 &1 day 14:28:30 &0:03:24 &9 days 15:51:21\\
\hline
20000 &500 &4800 &357 &0:00:05 &2:18:21 &0:00:35 &20:51:20\\
20000 &1000 &4800 &715 &0:00:12 &5:33:31 &0:01:40 &1 day 18:43:21\\
20000 &1500 &4800 &1072 &0:00:25 &9:44:07 &0:01:08 &2 days 20:08:20\\
20000 &2000 &4800 &1430 &0:00:31 &15:10:01 &0:01:29 &3 days 19:28:56\\
20000 &2500 &4800 &1788 &0:01:23 &1 day 12:24:25 &0:03:30 &4 days 20:53:53\\
20000 &3000 &4800 &2145 &0:01:50 &1 day 20:29:59 &0:03:44 &5 days 19:45:31\\
20000 &3500 &4800 &2503 &0:02:27 &2 days 5:40:09 &0:03:56 &6 days 19:44:54\\
20000 &4000 &4800 &2861 &0:03:03 &2 days 16:50:51 &0:03:46 &7 days 18:17:21\\
20000 &4500 &4800 &3218 &0:04:00 &3 days 3:35:02 &0:04:28 &8 days 19:49:26\\
20000 &5000 &4800 &3576 &0:04:52 &3 days 16:50:21 &0:04:44 &9 days 23:56:16\\
\hline
\label{tab:homoLog}
\end{tabular}
\end{center}

\paragraph{Logistic regression with heterogeneous data}

\begin{center}
\begin{tabular}{llllllll}
\hline
N &F &Cores &Space(GB) &TWalltime &TCompute &CWalltime &CCompute\\
\hline
50000 &2000 &800 &596 &0:00:56 &6:14:57 &0:09:25 &3 days 19:56:38\\
50000 &2000 &1600 &1192 &0:01:01 &12:28:12 &0:09:35 &7 days 19:00:17\\
50000 &2000 &2400 &1788 &0:00:58 &18:43:11 &0:09:35 &11 days 13:26:10\\
50000 &2000 &3200 &2384 &0:00:58 &1 day 1:09:09 &0:09:39 &15 days 10:33:19\\
50000 &2000 &4000 &2980 &0:01:23 &1 day 7:34:22 &0:09:49 &19 days 6:45:31\\
50000 &2000 &4800 &3576 &0:01:11 &1 day 13:51:15 &0:34:30 &77 days 5:23:50\\
50000 &2000 &5600 &4172 &0:01:29 &1 day 20:20:50 &0:34:38 &90 day 19:10:12\\
50000 &2000 &6400 &4768 &0:01:01 &2 days 2:55:20 &0:35:31 &103 days 19:09:22\\
50000 &2000 &7200 &5364 &0:01:14 &2 days 9:38:02 &0:10:26 &34 days 20:11:28\\
\hline
\label{tab:heteroLog}
\end{tabular}
\end{center}

\paragraph{Lasso with heterogeneous data}

\begin{center}
\begin{tabular}{llllllll}
\hline
N &F &Cores &Space(GB) &TWalltime &TCompute &CWalltime &CCompute\\
\hline
50000 &200 &800 &59 &0:00:12 &0:01:45 &0:00:37 &0:04:55\\
50000 &200 &1600 &119 &0:00:02 &0:03:31 &0:00:47 &0:10:56\\
50000 &200 &2400 &178 &0:00:02 &0:05:14 &0:01:14 &0:17:50\\
50000 &200 &3200 &238 &0:00:00 &0:07:00 &0:01:22 &0:25:24\\
50000 &200 &4000 &298 &0:00:04 &0:09:00 &0:01:36 &0:33:49\\
50000 &200 &4800 &357 &0:00:11 &0:10:25 &0:01:57 &0:43:29\\
50000 &200 &5600 &417 &0:00:10 &0:12:09 &0:02:07 &0:55:47\\
50000 &200 &6400 &476 &0:00:07 &0:13:48 &0:02:19 &1:04:51\\
50000 &200 &7200 &536 &0:00:09 &0:15:31 &0:02:39 &1:19:22\\
\hline
50000 &1000 &800 &298 &0:00:04 &0:33:28 &0:05:20 &2:58:02\\
50000 &1000 &1600 &596 &0:00:18 &1:06:33 &0:06:23 &6:00:37\\
50000 &1000 &2400 &894 &0:00:25 &1:39:50 &0:08:28 &9:04:25\\
50000 &1000 &3200 &1192 &0:00:09 &2:12:14 &0:08:34 &12:07:04\\
50000 &1000 &4000 &1490 &0:00:08 &2:46:27 &0:09:52 &15:13:18\\
50000 &1000 &4800 &1788 &0:00:21 &3:24:38 &0:13:28 &18:11:34\\
50000 &1000 &5600 &2086 &0:00:10 &3:50:29 &0:14:55 &21:25:49\\
50000 &1000 &6400 &2384 &0:00:06 &4:26:31 &0:16:11 &1 day 0:27:56\\
50000 &1000 &7200 &2682 &0:00:11 &4:56:57 &0:17:11 &1 day 3:34:19\\
\hline
\label{tab:heteroLasso}
\end{tabular}
\end{center}

\paragraph{SVM with homogeneous data}

\begin{center}
\begin{tabular}{llllllll}
\hline
N &F &Cores &Space(GB) &TWalltime &TCompute &CWalltime &CCompute\\
\hline
50000 &20 &48 &0 &0:00:01 &0:00:46 &0:02:45 &2:01:12\\
50000 &20 &96 &0 &0:00:01 &0:01:32 &0:02:47 &4:03:05\\
50000 &20 &144 &1 &0:00:02 &0:02:19 &0:02:49 &5:58:08\\
50000 &20 &192 &1 &0:00:02 &0:03:06 &0:02:45 &7:56:14\\
50000 &20 &240 &1 &0:00:02 &0:03:53 &0:02:51 &9:54:05\\
\hline
50000 &50 &48 &0 &0:00:03 &0:01:23 &0:05:14 &3:44:06\\
50000 &50 &96 &1 &0:00:03 &0:02:47 &0:05:19 &7:26:30\\
50000 &50 &144 &2 &0:00:03 &0:04:11 &0:05:25 &11:07:51\\
50000 &50 &192 &3 &0:00:07 &0:05:38 &0:05:25 &14:54:03\\
50000 &50 &240 &4 &0:00:03 &0:07:00 &0:05:25 &18:26:10\\
\hline
50000 &100 &48 &1 &0:00:05 &0:02:20 &0:09:28 &6:25:55\\
50000 &100 &96 &3 &0:00:05 &0:04:40 &0:09:56 &12:49:20\\
50000 &100 &144 &5 &0:00:05 &0:07:04 &0:09:45 &19:09:22\\
50000 &100 &192 &7 &0:00:06 &0:09:25 &0:09:53 &1 day 1:27:47\\
50000 &100 &240 &8 &0:00:05 &0:11:46 &0:10:06 &1 day 8:08:18\\
\hline
\label{tab:homoSVM}
\end{tabular}
\end{center}

\paragraph{Star data}

\begin{center}
\begin{tabular}{llllllll}
\hline
Cores &TWalltime &TCompute &CWalltime &CCompute\\
\hline
2500 &0:01:06 &11:35:25 &0:24:39 &31 days 19:59:13\\
3000 &0:00:49 &12:10:33 &0:21:43 &32 days 2:44:11\\
3500 &0:00:50 &12:17:27 &0:17:01 &30 days 7:56:19\\
4000 &0:00:45 &12:38:24 &0:29:53 &40 days 13:38:19\\
\hline
\end{tabular}
\end{center}

%

\end{document}